\documentclass[11pt, letterpaper]{article}
\usepackage{amssymb}
\usepackage{amsmath}
\usepackage{amsthm}
\usepackage[title]{appendix}
\usepackage{graphicx}
\usepackage{fullpage}
\usepackage{color}
\usepackage[letterpaper]{geometry}
\newtheorem{lemma}{Lemma}
\newtheorem{theorem}{Theorem}

\newenvironment{oldthm}[1]{\par\noindent{\bf Theorem #1:} \em \noindent}{\par}
\newenvironment{oldlem}[1]{\par\noindent{\bf Lemma #1:}
  \em \noindent}{\par}
\newenvironment{oldcor}[1]{\par\noindent{\bf Corollary #1:} \em \noindent}{\par}
\newenvironment{oldpro}[1]{\par\noindent{\bf Proposition #1:} \em \noindent}{\par}

\newcommand{\ethm}{\end{theorem}}
\newcommand{\commentout}[1]{}
\newcommand{\othm}[1]{\begin{oldthm}{\ref{#1}}}
\newcommand{\eothm}{\end{oldthm} \medskip}
\newcommand{\olem}[1]{\begin{oldlem}{\ref{#1}}}
\newcommand{\eolem}{\end{oldlem} \medskip}
\newcommand{\ocor}[1]{\begin{oldcor}{\ref{#1}}}
\newcommand{\eocor}{\end{oldcor} \medskip}
\newcommand{\opro}[1]{\begin{oldpro}{\ref{#1}}}
\newcommand{\eopro}{\end{oldpro} \medskip}

\usepackage{url}

\urldef{\mailt}\path|chhaya.dhingra@gmail.com|
\urldef{\mailv}\path|h.vandierendonck@qub.ac.uk|
\urldef{\mailk}\path|G.Karakonstantis@qub.ac.uk|
\urldef{\mailn}\path|d.nikolopoulos@qub.ac.uk|

\begin{document}


\title{Energy optimization of memory intensive parallel workloads\thanks{Regular Paper}}
\date{Queen's University of Belfast, Belfast, UK}


%
%
\author{Chhaya Trehan \thanks{\mailt (Corresponding Author)} \and Hans Vandierendonck\thanks{\mailv} \and Georgios Karakonstantis\thanks{\mailk}%
\and Dimitrios S. Nikolopoulos\thanks{\mailn}}
%


%
%


\begin{titlepage}
\maketitle

\begin{abstract}
Energy consumption is an important concern in modern multicore processors. The energy consumed during the execution of an application can be minimized by tuning the hardware state utilizing knobs such as frequency, voltage etc. The existing theoretical work on energy minimization using Global DVFS (Dynamic Voltage and Frequency Scaling), despite being thorough, ignores the energy consumed by the CPU on memory accesses and the dynamic energy consumed by the idle cores. This article presents an analytical model for the performance and the overall energy consumed  by the CPU chip on CPU instructions as well as the memory accesses without ignoring the dynamic energy consumed by the idle cores. We present an analytical framework around our energy-performance model to predict the operating frequencies for global DVFS that minimize the overall CPU energy consumption within a performance budget. Finally, we suggest a scheduling criteria for energy aware scheduling of memory intensive parallel applications.
\end{abstract}
\end{titlepage}

\section{Introduction}
While Silicon is available in abundance to build processors, the energy required to power them is not. Energy consumption and performance turn out to be the two most important and contradicting design criteria for the modern multicore processors~\cite{EnergyCost1, EnergyCost2, EnergyCost3}. The practice of dealing with the two contradicting goals by optimizing one while imposing a threshold on the other leads to two flavors of energy-performance optimization called the \emph{laptop problem} and the \emph{server problem}. In the laptop problem, the goal is to maximize the performance given a fixed energy budget and in the server problem, the goal is to minimize the energy consumption given a fixed performance budget~\cite{ServerLaptopProblem, ServerProblem}. We deal with the \emph{server problem} in this article.

 The energy consumed by a CMP(Chip Multi Processor) is an increasing function of the operating voltage and frequency of the chip and can be reduced by reducing one of them. Dynamic Voltage and Frequency Scaling (DVFS) is thus a popular energy minimization technique for multicore platforms.

While the raison d'etre of a multicore platform is to maximize the performance by maximizing the parallelism, the inherent parallelism of a workload is not easy to determine. It is thus becoming increasingly common to divide an application into a set of parallel tasks with precedence constraints so as to make the possible parallelism explicit. We consider a \emph{task dependency graph}~\cite{TaskGraphs1, TaskGraphs2} as our model for the workload. Often when we think of parallelism, we think of performance gains and we tend to ignore its ramifications on energy consumption. The fact that we can gain performance by increasing parallelism allows one to save energy by reducing frequency without violating a performance constraint. Gerards et al formalized the problem of energy consumption for task graphs in their recent article~\cite{ConvexAndScheduling} and studied the interplay between global DVFS and scheduling of parallel applications to minimize the CPU energy consumption. An interesting find of their work is that using a single clock frequency during the execution of an application does not lead to optimal energy consumption and they present an approach for varying the frequency during execution to minimize energy. The frequency is varied according to the variations in the amount of parallelism and a separate frequency is assigned to each number of active cores (parallelism).

The analytical model of~\cite{ConvexAndScheduling} for energy consumption and performance, however, completely ignores the energy consumed by the CPU while it waits for data accesses to the main memory. Since it does not account for the time overhead of the access latency of memory, it can lead to an imprecise estimate of the slack between the time to completion and the given performance budget. The CPU energy optimization techniques that save energy by decreasing the operating frequencies of the cores at the cost of an increased delay need to be tuned to account for the memory access latencies. Precisely accounting for the memory access delays of the application helps exploit the slack and avoids an over optimistic selection of operating frequencies. Another assumption in~\cite{ConvexAndScheduling} is that the frequency of the idle cores can be brought down to zero by techniques like clock gating. This is not always possible in reality, the idle cores can't be completely shut down and  do consume some  dynamic energy. In this article, we present a new model for the energy and performance of multicore systems that accounts for the energy consumed by the CMP while waiting for memory accesses in addition to the energy consumed on CPU instructions without ignoring the dynamic energy consumed by the idle cores. We provide an analytical framework around our energy-performance model to predict the operating frequencies for global DVFS to minimize the overall CPU energy consumption of a given application.

\paragraph{Related Work:}
A common approach to reduce the energy consumption of an application is to reduce the operating frequency of the cores~\cite{DVFS, DVFS2, DVFS3} which incurs a cost in terms of increased execution time. Most of the theoretical work on the energy-delay trade off deals with the \emph{local} DVFS~\cite{Local1, Local2, Local3}, where every core's voltage and frequency can be set separately. We study the problem of energy minimization under a performance constraint using global DVFS where the voltage and frequency are set for the entire chip. While local DVFS has more freedom in choosing clock frequencies and can therefore save more energy, it is not easy to implement~\cite{GlobalEasy}. Global DVFS being easier and cheaper to implement leads to much simpler and practical algorithms for choosing the frequencies for energy optimization problem. The relationship of parallelism with energy and performance was first studied by Sangyeun and Melhem in~\cite{ParallelizationEnergy}.  In their recent paper, Gerards et al~\cite{ConvexAndScheduling} show that using a single clock frequency during the execution of a \emph{parallel} application with precedence constraints does not lead to optimal energy consumption and present an approach for varying the frequency during execution to minimize energy. Li in his pioneering work~\cite{LiSchedulingAndEnergy} presents heuristic algorithms for energy optimization that treat scheduling and frequency selection as two independent subtasks performed one after the other. Further, Gerards et al~\cite{ConvexAndScheduling} show that the tasks of determining a schedule and frequencies that together minimize the energy consumption should not be considered separately and study the relation between the two. They define a scheduling criterion for energy optimization and show how to determine frequencies that minimize energy consumption. They characterize a schedule in terms of \emph{parallelism}, which gives for each number of cores the number of clock cycles 
for which exactly that many cores are active. Given a schedule, this model abstracts from the tasks and their precedence constraints and determines a clock frequency for each "number of active cores''.

\paragraph{Outline of the paper:}
The rest of the paper is organized as follows. Section~\ref{Model} describes the system, the application and the power model under consideration. Section~\ref{EnergyOptimizationProblem} formulates the energy optimization problem as a constrained convex optimization problem. Section~\ref{AffectMemoryAccesses} describes how memory accesses of a given application affect the optimal frequencies for energy optimization. In Section~\ref{SelectionScheduling}, we give analytical formulas for optimal frequencies for memory intensive workloads. Finally, Section~\ref{Conclusion} concludes the paper along with some research questions for future.
\section{Model}\label{Model}
Our application and system model is similar to the one presented in~\cite{ConvexAndScheduling}. 
The model presented here differs from that of~\cite{ConvexAndScheduling} in considering the memory accesses as being part of workload and the treatment of idle cores.
\paragraph{Application:}
We consider an application running on a multicore processor. The application itself consists of a set $T$ of $N$ tasks, denoted by ${T_1,.....T_N}$. We consider an overall deadline $t_{budget}$ for the entire application. A task $T_i$ is characterized by two attributes, namely: the \emph{compute workload:} $cw_i$ and the \emph{data workload} . The compute work load is the number of clock cycles required to perform the computations of the task. The data workload is the number of memory accesses a task has to make during its execution. We assume an application wide parameter called \emph{data to CPU} quotient $d$ which is the ratio of data to compute workloads of the application. It can be viewed as the number of memory accesses per CPU instruction cycle of the application. For a task $T_i$ with compute workload of $cw_i$, its data workload can be inferred as the product of $cw_i$ and $d$. We assume that the memory accesses of a task are distributed uniformly throughout the task.
 The application can be depicted as a labeled DAG (Directed Acyclic Graph) where nodes represent the tasks and the (\emph{Directed}) edges represent the precedence constraints (Figure~\ref{figure:TaskDependenyGraph} in Appendix). Each node carries a label depicting the CPU workload $cw_i$ of the associated task. 
 
\paragraph{Computing Platform:} The Application runs on a Chip Multiprocessor system with $M > 1$ homogeneous processing cores. All the cores have similar capabilities and run at the same frequency. Instead of using a single frequency throughout the application, we assume that the frequency can be changed at any time. We assume a frequency function $\phi:R^+ \to R^+$ which maps a given point in time to the frequency to be used at that time. Unlike~\cite{ConvexAndScheduling}, we do not assume that the frequency of inactive cores can be brought down to zero using clock gating. Therefore, one can not ignore the dynamic power consumption of inactive cores. We instead assume that the inactive cores run at the same frequency as the active cores, but their average activity factor is much less compared to the active cores. 

\paragraph{Power:}
As is common in literature, we consider two components of power, the \emph{dynamic power} and the \emph{static power}. Assuming $f$ is the frequency of all the cores at some time $t$, the dynamic power of an active core at time $t$ can be expressed as an increasing function of frequency as follows:
\begin{displaymath}
 p_{DynamicActive}(f) = c1f^{\alpha}
\end{displaymath}
The constant $c1 > 0$ is a characteristic of the computing platform and the exponent $\alpha$ is a constant($\geq 2$).
At any given point in time, an inactive core consumes relatively less dynamic power owing to its reduced activity factor. We model this difference in dynamic power of active and inactive cores by assuming that the constant $c1$ for inactive cores is less than the $c1$ for active cores.
Assuming $c1'$ to be the constant for inactive cores such that the ratio $K = \frac{c1'}{c1} < 1$, the dynamic power of an inactive core can be expressed as:
\begin{displaymath}
 p_{DynamicInactive}(f) = c1'f{\alpha}
\end{displaymath}
The static power which is a function of voltage can also be expressed as an affine function of frequency (since voltage and frequency are almost linearly related) as follows:
\begin{displaymath}
 p_{Static}(f) = c2f + c3
\end{displaymath}
At a given point in time, with $m$ active cores running at an operating frequency of $f$, the total power of the processor chip can be expressed as:
\begin{displaymath}
 p_m(f) = m c1f^{\alpha} + (M-m)c1'f^{\alpha} +  p_{Static}
\end{displaymath}
where $M$ is the total number of cores on the chip.
Expressing $c1'$ as $Kc1$, the equation for the total power with $m$ cores active at frequency $f$ is:
\begin{equation}\label{PowerEquation}
 p_m(f) = [m + k(M-m)]c1f^{\alpha} + c2f + c3
\end{equation}
This is a convex and increasing function in $f$. From this point on, we will denote $[m + k(M-m)]$ as $m'$ for the sake of brevity. 
Dividing equation~\ref{PowerEquation} on both sides by $f$ gives energy per CPU cycle which we will denote as  $\bar{p}_m$ henceforth

\begin{equation}\label{EnergyPerCycleEquation}
 \bar{p}_m(f) = m'c1f^{\alpha-1} + c2 + \frac{c3}{f}
\end{equation}

In~\cite{ConvexAndScheduling}, Gerards and others use the convex nature of the power function to prove that for an interval during which a constant number of cores are active, a constant frequency is optimal in terms of energy consumption. 

Before we go into the details of selecting the optimal frequencies in our model, we take a short diversion to understand what an interval $(t1,t2)$ in our model looks like and how the presence of  memory accesses during an interval change the dynamics of energy optimization. In any interval during the execution in our model, all the active cores are performing some memory accesses uniformly interleaved with the CPU instruction cycles. Therefore, not all of the CPU cycles produced during such an interval can be counted towards the work done (instructions) by the CPU. Moreover the time spent on memory accesses is independent of frequency whereas the time spent on executing the instructions can be increased (decreased) by decreasing (increasing) the frequency.  DVFS schemes for energy optimization exploit this ability to stretch an interval by decreasing the frequency to minimize energy consumption at the cost of increased delays. An interval with memory accesses can be thought of as composed of many springs with some rigid material placed between them. Applying a force (a change in frequency) can only compress or decompress the springs and the rigid material (memory accesses) does not yield at all to the changes in frequency. Only a portion of interval containing instruction cycles and memory accesses can be stretched by decreasing the frequency thus leading to a lesser potential for reduction in energy by decreasing the frequency. Coming back to the question of the optimal frequencies for an interval in our model during which a fixed number of cores are active, one can divide such an interval into many CPU only intervals separated by memory accesses stacked between them. Applying Lemma 1 of~\cite{ConvexAndScheduling} on each such interval, we deduce that we can use the same constant clock frequency during each such CPU only intervals. What we are left to decide is the frequencies to be used during memory accesses. We may wish to bring the frequency further down during these portions of the interval to get some energy savings. But the assumption that the memory accesses are uniformly interleaved throughout does not leave much room for reduction as the overhead of changing the frequency uniformly throughout the interval can offset the potential energy savings. We therefore stick to the idea of using a constant frequency for an interval during which a fixed number of cores are active.
\paragraph{Parallelism and Energy-Performance model:}
 The overall energy consumption of an application can be expressed in terms of the amount of parallelism. In interest of brevity, we refer the reader to go through~\cite{ConvexAndScheduling} to fully appreciate the concept of power modeling in terms of parallelism. For an application with $N$ tasks running on a processor with $M$ cores, its amount of parallelism for a given schedule can be defined formally  as a vector $[w_1 , w_2, ...w_m ...w_M ]$, where $w_m$ is the total number of CPU cycles for which exactly $m$ cores are active. Using the idea that a constant frequency for a fixed number of cores (parallelism) leads to an optimal energy consumption,the task of global DVFS for energy optimization is reduced to finding a vector $f = [f_1 , f_2 , ...f_m ...f_M ]$ of frequencies where $f_m$ is the optimal frequency to be used when $m$ cores are active. Energy consumed when $m$ cores are active can be expressed as the product of energy per cycle $\bar{p}_m$ from equation\ref{EnergyPerCycleEquation} and $w_m$.
 Thus the total energy consumption of the application without considering the memory accesses can be expressed as:
  \begin{equation}\label{EnergyNoMemory}
  E(f_1, f_2,......f_M) = \sum_{m=1}^{M}[\bar{p}_m(f_m)w_m]
 \end{equation} 
For a given amount of parallelism $w_m$ , $w_md$ accesses to memory are made, where $d$ is the application wide data to CPU workload ratio. The CPU keeps clocking at a frequency $f_m$ for the duration of these $w_md$ memory accesses. If $t_a$ is the latency of memory accesses, $(w_m d)t_a$ is the duration for which the CPU waits for memory accesses. The additional cycles expended per core on memory accesses for $w_m$ is thus $w_mdt_af_m$ . Replacing $w_m$ with $w_m +w_m d_taf_m$ in the energy equation~\ref{EnergyNoMemory} leads to a new energy equation that accounts for the CPU energy consumed not only on the actual CPU work done but also the addition clock cycles expended on waiting for the memory accesses. 

\begin{equation}\label{EnergyWithMemory}
  E_{total}(f_1, f_2,......f_M) = \sum_{m=1}^{M}[\bar{p}_m(f_m)(w_m + w_m dt_af_m)]
 \end{equation} 
  The time to completion of an application for a given schedule can also be expressed in terms of parallelism. The time taken when considering memory accesses has a frequency dependent and a constant component. The constant component of the time to completion is the memory overhead of the application. The time to completion in terms of parallelism is:
  \begin{equation}\label{TimeToCompletion}
  t_{completion}(f_1, f_2,......f_M) = \sum_{m=1}^{M}\frac{w_m}{f_m} +  \sum_{m=1}^{M}w_mdt_a
 \end{equation}
 
%
 \section{Energy Optimization}\label{EnergyOptimizationProblem}
Given an application and its schedule, the problem of energy optimization under a performance constraint can be formulated as one of finding an optimal set of frequencies, $f = [f_1 , f_2 , ...f_m ...f_M ]$  corresponding to the parallelism, $w =[w_1 , w_2, ...w_m ...w_M ]$. Denoting as $t_{budget}$, the deadline or the performance constraint of the application, the problem of energy minimization can be expressed as: 
\begin{equation}\label{OptimizationFirst}
\begin{aligned}
   &\underset{f_1, f_2, .....f_M}{\text{minimize}} \sum_{m=1}^{M}[\bar{p}_m(f_m)(w_m + w_m dt_af_m)] \\
   & \text{subject to} \quad \sum_{m=1}^{M}\frac{w_m}{f_m} + \sum_{m=1}^{M}w_mdt_a \leq t_{budget} \\
\end{aligned}
\end{equation}
Substituting the energy per cycle function $\bar{p}_m$ in equation~\ref{OptimizationFirst} with its expansion in equation~\ref{EnergyPerCycleEquation}, we get:
\begin{equation}\label{TotalEnergyFrequecy}
\begin{aligned}
   &\underset{f_1, f_2, .....f_M}{\text{minimize}} \sum_{m=1}^{M}[m'c1w_mdt_af_m^{\alpha} + m'c1w_mf_m^{\alpha-1}
 &+ c2w_mdt_af_m + c3\frac{w_m}{f_m} +c2w_m + c3w_mdt_a]\\
 & \text{subject to} \\
 &\qquad  \sum_{m=1}^{M}\frac{w_m}{f_m} + \sum_{m=1}^{M}w_mdt_a \leq t_{budget} \\
\end{aligned}
 \end{equation}
 Note that the decision variable $f = [f_1, f_2, .....f_M]$ can only take positive values, i.e. $f \in R_{+}^M$, thus making both the objective function and the constraint of the optimization problem(equation~\ref{TotalEnergyFrequecy}) convex~\cite{Boyd}. The solution to this problem is only a matter of typing in a few lines of code in any convex optimization solver.
%
 \section{Memory accesses and the optimal frequencies}\label{AffectMemoryAccesses}
 The main goal of this work is to study the effect of memory accesses on CPU energy consumption and how does their presence alter the optimal frequencies. 
 Gerards et al show in~\cite{ConvexAndScheduling} that for any given number of active cores $m$, the frequency $f_m$ is inversely proportional to $\alpha \sqrt{m}$. In this section, we will investigate how do the optimal frequencies relate to the memory intensity (data to CPU workload ratio, $d$) of an application and whether and how the relationship between optimal frequencies and the number of active cores change in the presence of memory accesses. Recall that in addition to  accounting for the energy consumption on memory accesses, we also account for the dynamic energy consumed by idle cores in our model.
 
 \begin{lemma}\label{UnconstrainedOptimization}
  On a given hardware platform the unconstrained minimizer $[f1,f2,....f_M]$ of energy is same for all the applications with a fixed data to CPU workload ratio.
 \end{lemma}
\begin{proof}
The unconstrained optimization problem is:
 \begin{equation}\label{UnconstrainedTotalEnergy}
\begin{aligned}
   &\underset{f_1, f_2, .....f_M}{\text{minimize}} E_{total} = \sum_{m=1}^{M}[m'c1w_mdt_af_m^{\alpha} + m'c1w_mf_m^{\alpha-1}
 + c2w_mdt_af_m + c3\frac{w_m}{f_m} +c2w_m + c3w_mdt_a]
 \end{aligned}
 \end{equation}
According to the optimality condition for unconstrained convex function~\cite{Boyd}, Gradient $\nabla(E_{Total}) = 0$ at the optimal point. Due to the separable nature of the objective function, one can easily get optimal $f_m$ for $m$ active cores by equating to zero the derivative of $m^{th}$ summand of the objective w.r.t. $f_m$.
 One can get $f_m$ by: 
 $$\frac{\partial E_{Total}}{\partial f_m} = 0$$
 \begin{displaymath}
 \begin{aligned}
  &\frac{\partial E_{Total}}{\partial f_m} = m'c1w_mdt_af_m^{\alpha-1} + m'c1w_m(\alpha -1)f_m^{alpha-2} 
  & + c2w_mdt_a -\frac{c3w_m}{f_m^2} 
  \end{aligned}
 \end{displaymath}
 thus, $f_m$ can be obtained by solving the following polynomial:
  \begin{equation}\label{UnconstrainedOptimizer}
 \begin{aligned}
  m'c1dt_a\alpha f_m^{\alpha +1} + m'c1(\alpha-1)f_m^{\alpha} + c2dt_af_m^2 -c3 = 0
  \end{aligned}
 \end{equation}
 The workload term $w_m$ gets canceled out. The only characteristic of the application in this polynomial is its memory intensity $d$, all other terms in the coefficients are characteristics of the underlying hardware on which the application runs. Hence all applications with a given memory intensity $d$ running on a given hardware platform will have same optimal value for $f_m$. The same holds true for all the other frequencies.
\end{proof}
A deadline constraint can change the optimal frequencies if the unconstrained minimizer does not meet the deadline. The frequencies however should have some relationship to each other based on their relative number of active cores (parallelization).

\begin{theorem}\label{TheoremRange}
It holds for every pair $n,m \in \{1,2,....M\}$ such that $m\geq n$ and $w_m,w_n>0$ that:
 \begin{enumerate}
 \item for an optimal solution $f=[f1,f2,....f_M]$ to the constrained energy optimization problem (equation~\ref{TotalEnergyFrequecy}), $\frac{f_m}{f_n}$ lies in the interval $[\sqrt[\alpha]{\frac{n'}{m'}}, 1]$.
  \item for an optimal solution $f=[f1,f2,....f_M]$ to the constrained energy optimization problem without the static energy,
  $\frac{f_m}{f_n}$ lies in the interval $[\sqrt[\alpha]{\frac{n'}{m'}},\sqrt[\alpha+1]{\frac{n'}{m'}}]$.
 \end{enumerate}
\end{theorem}
\begin{proof}
 For an arbitrary pair $n,m \in \{1,2,....M\}$ with $w_m,w_n>0$, both $f_m$ and $f_n$ are positive, hence there exists a positive constant $x$ such that $f_m = xf_n$.
 Let $t_{n,m}$ be the total time for which $n$ or $m$ cores are active. One can increase one of $f_n$ or $f_m$ and decrease the other such that the total time $t_{n,m}$, remains constant. 
 The total time $t_{n,m}$ can be expressed in terms of $f_n$ and $f_m$ as:
 \begin{displaymath}
  t_{n,m} = \frac{w_m}{f_m} + \frac{w_n}{f_n} + w_mdt_a + w_ndt_a
 \end{displaymath}
Substituting $t'_{n,m}$ for $t_{n,m} - w_mdt_a - w_ndt_a$ each of $f_m$ and $f_n$ can be expressed as a function of $x$ as follows:
\begin{equation}
 f_n = \frac{w_n + \frac{w_m}{x}}{t'{n,m}}
\end{equation}
\begin{equation}
 f_m = \frac{w_m + w_nx}{t'{n,m}}
\end{equation}
Let $E_{n,m}$ be the total energy consumed by the CPU during the time $t_{n,m}$.
\begin{displaymath}
\begin{aligned}
 &E_{n,m} = [m'c1f_m^{\alpha-1} + c2 + \frac{c3}{f_m}][w_m + w_mdt_af_m] 
           + [n'c1f_n^{\alpha-1} + c2 + \frac{c3}{f_n}][w_n + w_ndt_af_n]
\end{aligned}
\end{displaymath}
Rearranging the terms, we get:
\begin{equation}\label{EnergyNMParts}
\begin{aligned}
 &E_{n,m} = \underbrace{m'c1w_mdt_af_m^{\alpha} + n'c1w_ndt_af_n^{\alpha}}_\text{Memory Accesses} 
             + \underbrace{ m'c1w_mf_m^{\alpha-1} + n'c1w_nf_n^{\alpha-1}}_\text{CPU instructions}\\
           &\qquad\ \ + \underbrace{c2w_mdt_af_m + + c_2w_ndt_af_n + c2w_m + c2w_n + c3t_{n,m}}_\text{Static Energy}
\end{aligned}
\end{equation}
As highlighted in equation~\ref{EnergyNMParts}, the total energy $E_{n,m}$ is composed of three components, the first two terms constitute the CPU energy consumed on the memory accesses, the next two terms constitute the CPU energy consumed on CPU instructions and rest of the terms constitute the static energy which is independent of number of active cores.
Note that each of the three components is convex and nondecreasing in $[f_m,f_n]$, and $f_m$ and $f_n$ are convex~\cite{Boyd} in $x$, thus making these components convex functions of $x$.
We can therefore find the optimal ratio $x = \frac{f_m}{f_n}$ which minimizes the total energy $E_{n,m}$, by evaluating $\frac{dE_{n,m}}{dx} = 0$

\begin{displaymath}
\begin{aligned}
& \frac{dE_{n,m}}{dx} = \frac{w_nw_m}{t'_{n.m}}[m'c1dt_a\alpha f_m^{\alpha -1} + m'c1(\alpha-1)f_m^{\alpha-2} + c2dt_a]\\
& -\frac{w_nw_m}{t'_{n.m}}[\frac{n'c1dt_a\alpha f_n^{\alpha -1}}{x^2} + \frac{m'c1(\alpha-1)f_m^{\alpha-2}}{x^2} + \frac{c2dt_a}{x^2}]
 \end{aligned}
\end{displaymath}
Making $\frac{dE_{n,m}}{dx} = 0$, we get:
\begin{equation}\label{MinimzerParts}
\begin{aligned}
&\underbrace{c1dt_a\alpha [m'f_m^{\alpha-1} - n'\frac{f_n^{\alpha-1}}{x^2}]}_\text{Memory Accesses}
&+\underbrace{c1(\alpha-1)[m'f_m^{\alpha -2} -n'\frac{f_n^{\alpha-2}}{x^2}]}_\text{CPU instructions}\\
&+\underbrace{c1dt_a\alpha[1-\frac{1}{x^2}]}_\text{Static Energy} = 0
 \end{aligned}
\end{equation}
Making the first component of~\ref{MinimzerParts} labeled `Memory Accesses' zero, we get $x$ that minimizes the energy spent on memory accesses. Similarly making the  second and the third components of equation~\ref{MinimzerParts} labeled `CPU instructions' and `Static Energy' zero respectively, we get $x$ that minimizes the energy spent on CPU instructions and $x$ that minimizes the static energy respectively.
Let us represent by $x_{mem}$, $x_{CPU}$ and $x_{static}$, the values of $x$ that minimize the energy consumed on memory accesses, CPU instructions and static energy respectively.
Thus we have, $x_{mem} = \sqrt[\alpha + 1]{\frac{n'}{m'}}$, $x_{CPU} = \sqrt[\alpha]{\frac{n'}{m'}}$, $x_{static} = 1$ .
Since each of the three components of $E_{n,m}$ are convex in $x$ with each having a unique minimum different from the other two, the minimizer of the total energy lies somewhere in the interval $[min\{x_{CPU}, x_{mem}, x_{static}\}, max\{x_{CPU}, x_{mem}, x_{static}\}]$~\cite{Boyd}. Since we assumed that $m \geq n$, the minimizer of total energy lies in the range $[\sqrt[\alpha]{\frac{n'}{m'}}, 1]$.
Dropping the component for static energy in~\ref{MinimzerParts}, we get the sum of the  dynamic energy consumed by the CPU on memory accesses and CPU instructions respectively and its minimizer lies in the range $[\sqrt[\alpha]{\frac{n'}{m'}}, \sqrt[\alpha + 1]{\frac{n'}{m'}}]$.
\end{proof}
The implication of theorem~\ref{TheoremRange} is that the presence of memory accesses affects the frequency selection in two ways. First the memory accesses add to the CPU energy consumption a dynamic component which increases more sharply with frequency than the dynamic energy consumed on CPU instructions, thus pushing the optimal ratio further away. Second, it adds to the total energy consumption a static component that varies linearly with frequency and this component is minimized when $\frac{f_m}{f_n} = 1$. The minimizer of the overall energy tends to be closest to the minimizer of the most dominating component in the mix. It is well known that the static energy despite being an unavoidable portion of the overall energy is a much smaller component of the total energy compared to the dynamic energy.
In rest of this section, we focus our attention on the dynamic energy component of the overall energy.

\begin{lemma}\label{Relation}
for an optimal solution $f=[f1,f2,....f_M]$ to the constrained energy optimization problem(equation~\ref{TotalEnergyFrequecy}) without the static energy, the following holds for every pair $n,m \in \{1,2,....M\}$ with $w_m,w_n>0$:\\
 \begin{displaymath}
  \sqrt[\alpha]{m'[\alpha -1 + \alpha dt_af_m]}f_m =  \sqrt[\alpha]{n'[\alpha -1 + \alpha dt_af_n]}f_n
 \end{displaymath}
\end{lemma}
 Refer the Appendix for the proof. \\
This Lemma shows the relationship between the frequencies for two different parallel regions of a given schedule of an application. This is in contrast to the corresponding relationship in~\cite{ConvexAndScheduling}, which is, $\sqrt[\alpha]{n}f_n = \sqrt[\alpha]{m}f_m$. 

Having a relationship between the optimal frequencies for different parallel regions of a schedule, the next natural step is to be able to analytically relate an optimal frequency for a parallel region to the optimal frequency of the serial region.  Lemma~\ref{CubicEquation} gives such a relation for $\alpha=2$. 

\begin{lemma}\label{CubicEquation}
 For $\alpha=2$, the ratio $x_m= \frac{f_m}{f_1}$, of the optimal frequency $f_m$ for a parallel region of the schedule with $m$ active cores and the optimal frequency $f_1$ for the serial region is a solution to the following cubic equation:
 \begin{displaymath}
   \frac{m'}{1'}2dt_af1x_m^3 + \frac{m'}{1'}x_m^2 -(2dt_af1 + 1) = 0
 \end{displaymath}
 where $1'$ is a constant equal to $KM + (1-K)$
\end{lemma}
Refer the appendix for the proof.\\
Note that it is possible for a schedule to have no serial region at all ($w_1 =0$). The purpose of expressing $x_m$ in terms of $f_1$ and $m'$ is to help understand by how much does the frequency for a given parallelization differ from $f_1$. One can think of $f_1$ as a \emph{reference frequency} for a given hardware and application combination such that each of the optimal parallel frequencies $f_m$ is related to $f_1$ by a multiplicative factor $x_m$.
Coming back to equation~\ref{RelativeCritical}, the coefficient of the cubic term is a product of the parallelization $m'$, the memory characteristic $2dt_a$ of the workload and the critical serial frequency $f_1$ and the memory intensity of the application. For the sake of analysis, we call the term $2dt_af1$, the memory overload factor. 
Plotting $\frac{f_m}{f_1}$ against the memory overload factor, $2dt_af_1$(Figure~\ref{figure: ReductionRatio} in Appendix), one can observe that as $2dt_af_1$  changes from $0$ to $1$, the optimal ratio changes very quickly and attains the mid point of $\sqrt[2]{\frac{1}{m}}$ and $\sqrt[3]{\frac{1}{m}}$ and then it changes more slowly and later becomes almost constant close to $\sqrt[3]{\frac{1}{m}}$. So, as memory overhead increases, the optimal frequency for $m$ active cores tends to be inversely proportional to $\sqrt[3]{m}$. Without accounting for the memory accesses or for CPU intensive applications on the other hand the optimal frequency for $m$ active cores is inversely proportional to $\sqrt[2]{m}$. Thus accounting for memory accesses does not allow as much reduction in frequencies for parallel regions as predicted by the model in~\cite{ConvexAndScheduling}. 
This confirms that the energy savings predicted by~\cite{ConvexAndScheduling} are over optimistic especially in the case of memory intensive applications (with a high memory intensity, $d$) running on a slow hardware (with a high access 
delay, $t_a$) on a tight performance budget (with a high critical frequency, $f_1$).
In general, from theorem~\ref{TheoremRange} and generalizing the above exposition, it can be established that the optimal frequency for $m$ active cores for a memory intensive application (with $\alpha dT_af_1$ sufficiently large) is inversely proportional to $\sqrt[alpha+1]{m}$.

\section{Frequency Selection and Scheduling criteria for memory intesive workloads}\label{SelectionScheduling}
 In this section, we look at the problem of frequency selection for a memory intensive application analytically to investigate the relationship between optimal frequencies and the distribution of workload or schedule of the application. As demonstrated in Section~\ref{AffectMemoryAccesses}, the optimal frequency $f_m$ for a memory intensive application considering only the dynamic energy is inversely proportional to $\sqrt[\alpha +1]{m}$ and one can consider the existence of a reference frequency such that each of the optimal frequencies can be expressed as a product of $\frac{1}{\sqrt[\alpha +1]{m}}$ and the reference frequency. Let us denote the reference frequency by $f'$. Substituting $f_m$ with $\frac{f'}{\sqrt[\alpha +1]{m}}$ in the dynamic energy only part of optimization problem given in equation~\ref{TotalEnergyFrequecy}, we get
\begin{displaymath}
\begin{aligned} 
  &\underset{f'}{\text{minimize}}[\sum_{m=1}^{M}\pi_{m'}w_m c1dt_af'^{\alpha} + \sum_{m=1}^{M}\pi_{m'}^2 w_m c1f'^{\alpha-1}]\\
 & \text{subject to}\\
 & \qquad \qquad f' \geq \frac{\sum_{m=1}^{M}\pi_{m'} w_m}{t_{budget} - \sum_{m=1}^{M}w_mdt_a}
\end{aligned}
\end{displaymath}
where $\pi_{m'} =\sqrt[\alpha+1]{m'}$ .
One can apply the KKT conditions~\cite{Boyd} on the above convex optimization problem to find an analytical formula for the optimal reference frequency. The optimal reference frequency for dynamic energy is:
\begin{equation}
f' = \frac{\sum_{m=1}^{M}\pi_{m'}w_m}{t_{budget} - \sum_{m=1}^{M}w_mdt_a}
\end{equation}
Refer Lemma~\ref{DynamicOptimizer} in the appendix for the proof.

If one were to minimize the total energy consumption including the static energy, the optimal ratio $\pi_{m'}$ would lie between $\sqrt[\alpha+1]{m'}$ and $1$ (Theorem\ref{TheoremRange}). Since static energy is only a small portion of the overall energy, we can assume $\pi_{m'} \approx  \sqrt[\alpha+1]{m'}$. Substituting  $\pi_{m'}$ with $\sqrt[\alpha+1]{m'}$ in equation~\ref{TotalEnergyFrequecy} and applying KKT conditions we get:
\begin{displaymath}
f' = max(\text{Unconstrained optimizer of }  tEnergy(f'),  \frac{\sum_{m=1}^{M}\pi_{m'}w_m}{t_{budget} - \sum_{m=1}^{M}w_mdt_a})
\end{displaymath}
where $\pi_{m'} =\sqrt[\alpha+1]{m'}$ and \\
\begin{displaymath}
\begin{aligned}
 &tEnergy(f') = \sum_{m=1}^{M}\pi_{m'}w_m c1dt_af'^{\alpha} + \sum_{m=1}^{M}\pi_{m'}^2 w_m c1f'^{\alpha-1} \\
& \qquad + c2dt_a\sum_{m=1}^{M} w_m \frac{f'}{\pi_{m'}}  +  c3\sum_{m=1}^{M}\frac{w_m\pi_{m'}}{f'}  + c2\sum_{m=1}^{M}w_m + c3dt_a\sum_{m=1}^{M}w_m
 \end{aligned}
 \end{displaymath}
 Refer Lemma~\ref{TotalOptimizer} in the appendix for the proof.
 
\paragraph{Schedule}
We should make an observation about the \emph{makespan}~\cite{Brucker} (commonly used as a performance measure of scheduling algorithms). Gerards et al show in~\cite{ConvexAndScheduling} that defining the makespan of an application as the number of CPU cycles required, instead of the time required to run the application ($S = \sum_{m=1}^{M}w_m$)  gives us a definition independent of frequency. Accounting for the CPU cycles produced during memory accesses, the definition of makespan becomes $\sum_{m=1}^{M}w_m + w_mdt_af_m$ which is not independent of frequency.
\paragraph{Scheduling Criteria:} The reference frequency is the largest of all the optimal frequencies and the dynamic energy is an increasing function of frequencies. Thus minimizing the minimum allowed reference frequency minimizes the dynamic energy. Therefore, when looking for a schedule that minimizes the energy under a performance constraint, one should pick the one which minimizes the minimum allowed reference frequency.
 Therefore a schedule that minimizes the following quantity is optimal in terms of energy consumption:
 \begin{displaymath}
   \frac{\bar{S}}{t_{budget} - Sdt_a}
 \end{displaymath}
  where $\overline{S} =\sum_{m=1}^{M}\pi_{m'}w_m$ is the weighted sum of the parallelism vector, $w=[w_1, w_2,....w_M]$  and $S = \sum_{m=1}^{M}w_m$ is the sum of the parallelism vector.
  Note that the above scheduling criteria is different from the one suggested in~\cite{ConvexAndScheduling} which suggests minimizing  $\overline{S} =\sum_{m=1}^{M}\pi_{m'}w_m$ , which the authors define as the weighted makespan. Traditionally, scheduling algorithms minimize the makespan~\cite{MinMakespan1, MinMakespan2, Pinedo} and Gerards et al compare the traditional performance measure which is makespan with their scheduling criteria to suggest that minimizing the weighted makespan minimizes the energy consumption. The scheduling criteria we suggest here in contrast combines the criteria of~\cite{ConvexAndScheduling} with that of the traditional measure.

\section{Conclusion and Future Work}\label{Conclusion}
We have presented a comprehensive study on the frequency selection and workload distribution (scheduling) for energy optimization of memory intensive parallel workloads. In this work, we assume that all the tasks in a workload have the same data to CPU workload ratio. In future, we plan to extend this work by allowing the tasks to have different memory intensities.

%

\pagebreak
\bibliographystyle{abbrv}
\bibliography{EnergyEfficiency.bib} 

\pagebreak

\appendix
\renewcommand{\thesection}{\Alph{section}}%
\section{Figures}
\begin{figure}[h!]
  \begin{center}
   \includegraphics[scale=0.30]{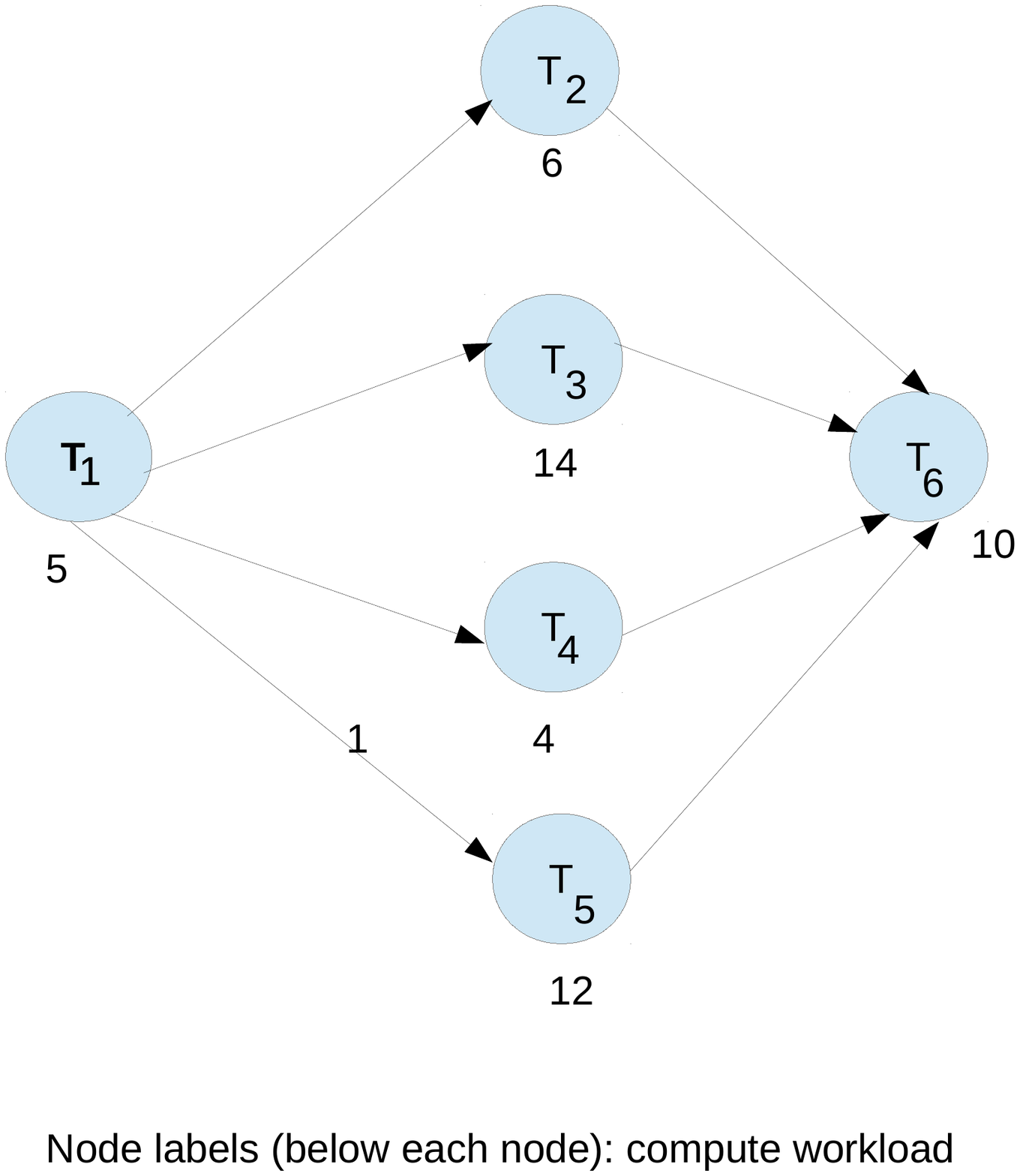}
  \end{center}
  \vspace*{-15mm}
  \caption{A task dependency graph}
   \label{figure:TaskDependenyGraph}
 \end{figure}
 
 \begin{figure}[h!]
  \begin{center}
   \includegraphics[scale=0.21]{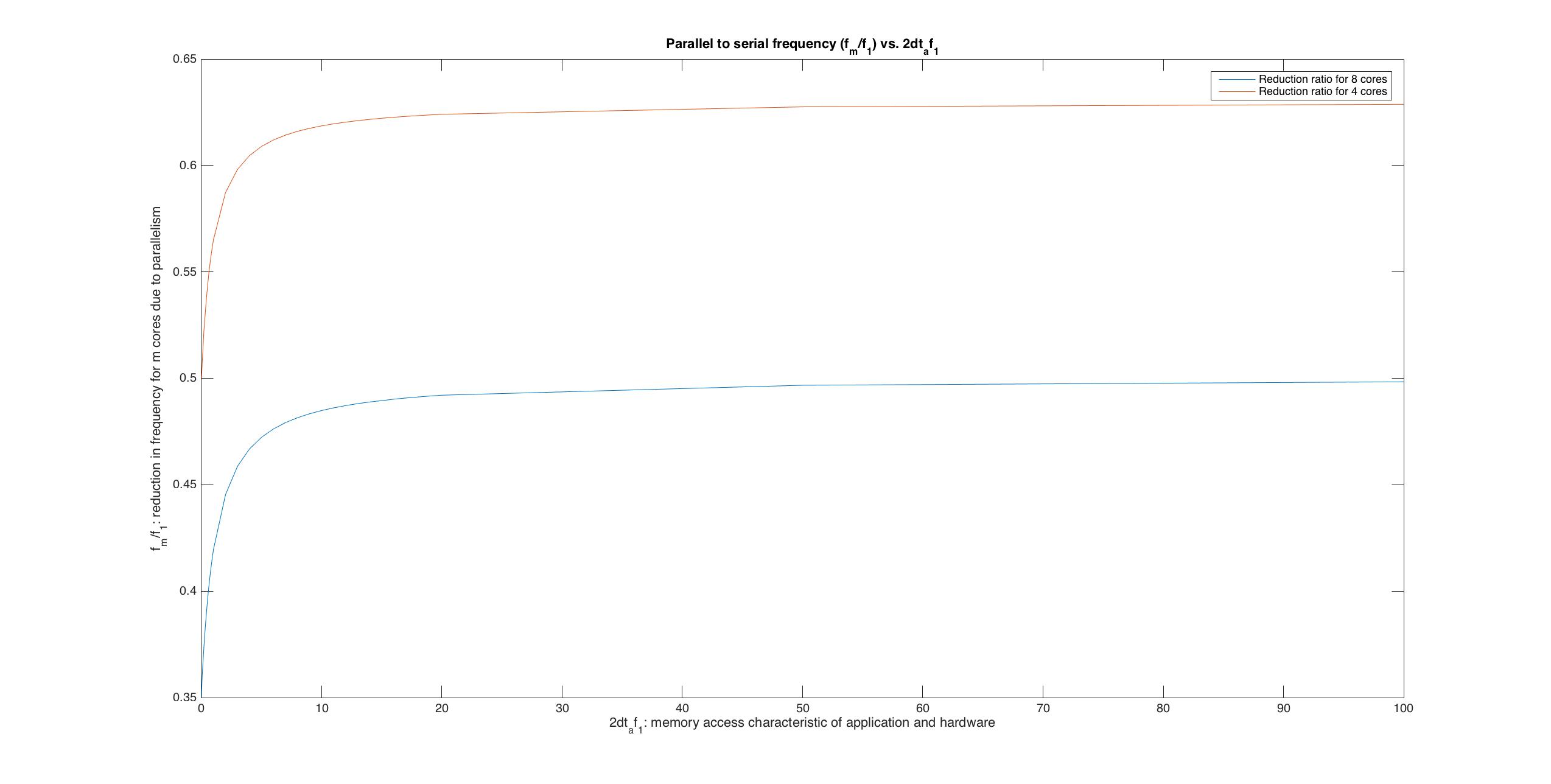}
  \end{center}
  \vspace*{-15mm}
   \caption{Optimal ratio vs memory overload}
   \label{figure: ReductionRatio}
 \end{figure}

\section{Proofs}

\olem{Relation}
for an optimal solution $f=[f1,f2,....f_M]$ to the constrained energy optimization problem(equation~\ref{TotalEnergyFrequecy}) without the static energy, the following holds for every pair $n,m \in \{1,2,....M\}$ with $w_m,w_n>0$:\\
 \begin{displaymath}
  \sqrt[\alpha]{m'[\alpha -1 + \alpha dt_af_m]}f_m =  \sqrt[\alpha]{n'[\alpha -1 + \alpha dt_af_n]}f_n
 \end{displaymath}
\eolem

\begin{proof}
 To prove the relation between optimal $f_m$ and $f_n$ for dynamic energy, we take the `Memory Accesses' and the `CPU instructions' components of the equation~\ref{MinimzerParts}.
 \begin{displaymath}
   \begin{aligned}
   &\frac{dE_{n,m}}{dx} =
 &\underbrace{c1dt_a\alpha [m'f_m^{\alpha-1} - n'\frac{f_n^{\alpha-1}}{x^2}]}_\text{Memory Accesses}
 &+\underbrace{c1(\alpha-1)[m'f_m^{\alpha -2} -n'\frac{f_n^{\alpha-2}}{x^2}]}_\text{CPU instructions}\\
  \end{aligned}
 \end{displaymath} 
 $\frac{dE_{n,m}}{dx} = 0$ at the optimal point.
  \begin{displaymath}
   \begin{aligned}
   &\frac{dE_{n,m}}{dx} = \\
  & c1dt_a\alpha [m'f_m^{\alpha-1} - n'\frac{f_n^{\alpha-1}}{x^2}] +c1(\alpha-1)[m'f_m^{\alpha -2} -n'\frac{f_n^{\alpha-2}}{x^2}] = 0 \implies\\
 &m'[\alpha -1 + \alpha dt_af_m]f_m^{\alpha-2} = \frac{f_n^{\alpha-2}}{x^2}n'[\alpha -1 + \alpha dt_af_n] \implies\\
  \end{aligned}
 \end{displaymath}
 \begin{equation}\label{RelationEquation}
   \frac{f_m}{f_n} = \frac{\sqrt[\alpha]{n'[\alpha -1 + \alpha dt_af_n]}}{\sqrt[\alpha]{m'[\alpha -1 + \alpha dt_af_m]}}
 \end{equation}
\end{proof}

\olem{CubicEquation}
 For $\alpha=2$, the ratio $x_m= \frac{f_m}{f_1}$, of the optimal frequency $f_m$ for a parallel region of the schedule with $m$ active cores and the optimal frequency $f_1$ for the serial region is a solution to the following cubic equation:
 \begin{displaymath}
   \frac{m'}{1'}2dt_af1x_m^3 + \frac{m'}{1'}x_m^2 -(2dt_af1 + 1) = 0
 \end{displaymath}
 where $1'$ is a constant equal to $KM + (1-K)$
\eolem
\begin{proof}
 From equation~\ref{RelationEquation} we have
 \begin{displaymath}
  \frac{m'}{n'}(\frac{f_m}{f_n})^\alpha = \frac{\alpha -1 + \alpha dt_af_n}{\alpha -1 + \alpha dt_af_m}
 \end{displaymath}
 Replacing $f_m$ by $x_mf_n$ throughout, where $x_m = \frac{f_m}{f_1}$ we get
 \begin{displaymath}
 \begin{aligned}
  &\frac{m'}{n'}x_m^\alpha = \frac{\alpha -1 + \alpha dt_af_n}{\alpha -1 + \alpha dt_ax_mf_n} \implies\\
 & m'\alpha dt_af_nx_m^{alpha+1} + m'(\alpha-1)x_m^\alpha = n'(\alpha - 1 + \alpha dt_af_n) \implies \\
 &\frac{m'}{n'}\alpha dt_af_nx_m^{alpha+1} + \frac{m'}{n'}(\alpha-1)x_m^\alpha -(\alpha dt_af_n + \alpha -1) = 0
  \end{aligned}
 \end{displaymath}
 We now have the ratio of two optimal frequencies expressed in terms of one of the two frequencies $(f_n)$ and their relative parallelization $(\frac{m'}{n'})$ and the memory intensity of the application in question. 
 Substituting $f_n$ with $f_1$, $n'$ with $1'$ for serial region,  where $1'= KM + 1-K$  and $\alpha$ with 2, we get 
 \begin{equation}\label{RelativeCritical}
  \frac{m'}{1'}2dt_af1x_m^3 + \frac{m'}{1'}x_m^2 -(2dt_af1 + 1) = 0
  \end{equation}
\end{proof}

\begin{lemma}\label{DynamicOptimizer}
The optimal reference frequency for minimizing the dynamic energy consumption of memory intensive applications  is:
\begin{displaymath}
f' = \frac{\sum_{m=1}^{M}\pi_{m'}w_m}{t_{budget} - \sum_{m=1}^{M}w_mdt_a}
\end{displaymath}
where $\pi_{m'} =\sqrt[\alpha+1]{m'}$.
\end{lemma}
\begin{proof}
As explained in Section~\ref{SelectionScheduling}, the optimization problem for the dynamic energy, considering $f_m = \frac{f'}{\sqrt[\alpha+1]{m'}}$ is:

\begin{displaymath}
\begin{aligned} 
  &\underset{f'}{\text{minimize}}[\sum_{m=1}^{M}\pi_{m'}w_m c1dt_af'^{\alpha} + \sum_{m=1}^{M}\pi_{m'}^2 w_m c1f'^{\alpha-1}]\\
 & \text{subject to}\\
 & \qquad \qquad f' \geq \frac{\sum_{m=1}^{M}\pi_{m'}w_m}{t_{budget} - \sum_{m=1}^{M}w_mdt_a}\\
 &\qquad \qquad f' \geq 0
\end{aligned}
\end{displaymath}
where $\pi_{m'} =\sqrt[\alpha+1]{m'}$.

Denoting the objective function by $dEnergy(f')$, the \emph{Lagrangian}~\cite{Boyd} for this optimization problem is:

\begin{displaymath}
L(f', \lambda_1, \lambda_2) = dEnergy(f') +  \lambda_1(\frac{\sum_{m=1}^{M}\pi_{m'}w_m}{t_{budget} - \sum_{m=1}^{M}w_mdt_a} -f') -\lambda_2f'
\end{displaymath}
where $\lambda_1$ and $\lambda_2$ are the Lagrange multipliers for the inequality constraints.
The differential of the Lagrangian w.r.t. $f'$ is:
\begin{displaymath}
     \frac{d L(f', \lambda_1, \lambda_2)}{df'} =  \frac{ddEnergy(f')}{df'}  -\lambda_1 -\lambda_2
\end{displaymath}
Applying KKT condition  $\frac{d L(f', \lambda_1, \lambda_2)}{df'} = 0$ at the optimal point, we get:
\begin{equation}\label{LambdaSum1}
    \lambda_1 + \lambda_2  =  \frac{ddEnergy(f')}{df'}  
\end{equation}
Note that $\lambda_1 > 0$, since $dEnergy(f')$ is an increasing function of $f'$.
Applying the complementary slackness condition on the second constraint ($f' \geq 0$), we get, $\lambda_2f' = 0$, $f'$ can't be zero, therefore we get  $\lambda_2 = 0$. It follows thus from equation~\ref{LambdaSum1}  that:
\begin{displaymath}
\lambda_1 =  \frac{ddEnergy(f')}{df'}  
\end{displaymath}
Applying complementary slackness on the the first constraint, we get $\lambda_1(\frac{\sum_{m=1}^{M}\pi_{m'}w_m}{t_{budget} - \sum_{m=1}^{M}w_mdt_a} -f') = 0$, which further implies that $\frac{\sum_{m=1}^{M}\pi_{m'}w_m}{t_{budget} - \sum_{m=1}^{M}w_mdt_a} -f' = 0$.
Thus the optimal point is:
\begin{displaymath}
\begin{aligned}
&\lambda_1 =  \frac{ddEnergy(f')}{df'} \\
&\lambda_2 = 0\\
&f' = \frac{\sum_{m=1}^{M}\pi_{m'}w_m}{t_{budget} - \sum_{m=1}^{M}w_mdt_a}
\end{aligned}
\end{displaymath}
\end{proof}

\begin{lemma}\label{TotalOptimizer}
The optimal reference frequency for minimizing the overall energy consumption of memory intensive applications  is:
\begin{displaymath}
f' = max(\text{Unconstrained optimizer of }  tEnergy(f'),  \frac{\sum_{m=1}^{M}\pi_{m'}w_m}{t_{budget} - \sum_{m=1}^{M}w_mdt_a})
\end{displaymath}
where $\pi_{m'} =\sqrt[\alpha+1]{m'}$ and \\
\begin{displaymath}
\begin{aligned}
 &tEnergy(f') = \sum_{m=1}^{M}\pi_{m'}w_m c1dt_af'^{\alpha} + \sum_{m=1}^{M}\pi_{m'}^2 w_m c1f'^{\alpha-1} \\
& \qquad + c2dt_a\sum_{m=1}^{M} w_m \frac{f'}{\pi_{m'}}  +  c3\sum_{m=1}^{M}\frac{w_m\pi_{m'}}{f'}  + c2\sum_{m=1}^{M}w_m + c3dt_a\sum_{m=1}^{M}w_m
 \end{aligned}
 \end{displaymath}
\end{lemma}
\begin{proof}
The optimization problem for the total energy, considering $f_m = \frac{f'}{\sqrt[\alpha+1]{m'}}$ is:

\begin{equation}\label{TotalEnergyFrequecy2}
\begin{aligned}
   &\underset{f'}{\text{minimize}} [\sum_{m=1}^{M}\pi_{m'}w_m c1dt_af'^{\alpha} + \sum_{m=1}^{M}\pi_{m'}^2 w_m c1f'^{\alpha-1} \\
 &+ c2dt_a\sum_{m=1}^{M} w_m \frac{f'}{\pi_{m'}}  +  c3\sum_{m=1}^{M}\frac{w_m\pi_{m'}}{f'}  + c2\sum_{m=1}^{M}w_m + c3dt_a\sum_{m=1}^{M}w_m]\\
 & \text{subject to}\\
 & \qquad \qquad f' \geq \frac{\sum_{m=1}^{M}\pi_{m'}w_m}{t_{budget} - \sum_{m=1}^{M}w_mdt_a}\\
 &\qquad \qquad f' \geq 0
\end{aligned}
 \end{equation}

where $\pi_{m'} =\sqrt[\alpha+1]{m'}$.

Denoting the objective function by $tEnergy(f')$, the \emph{Lagrangian}~\cite{Boyd} for this optimization problem is:

\begin{displaymath}
L(f', \lambda_1, \lambda_2) = tEnergy(f') +  \lambda_1(\frac{\sum_{m=1}^{M}\pi_{m'}w_m}{t_{budget} - \sum_{m=1}^{M}w_mdt_a} -f') -\lambda_2f'
\end{displaymath}
where $\lambda_1$ and $\lambda_2$ are the Lagrange multipliers for the inequality constraints.
The differential of the Lagrangian w.r.t. $f'$ is:
\begin{displaymath}
     \frac{d L(f', \lambda_1, \lambda_2)}{df'} =  \frac{dtEnergy(f')}{df'}  -\lambda_1 -\lambda_2
\end{displaymath}
Applying KKT condition  $\frac{d L(f', \lambda_1, \lambda_2)}{df'} = 0$ at the optimal point, we get:
\begin{equation}\label{LambdaSum2}
    \lambda_1 + \lambda_2  =  \frac{ddEnergy(f')}{df'}  
\end{equation}
Applying the complementary slackness condition on the second constraint ($f' \geq 0$), we get, $\lambda_2f' = 0$, $f'$ can't be zero, therefore we get  $\lambda_2 = 0$. It follows thus from equation~\ref{LambdaSum2}  that:
\begin{displaymath}
\lambda_1 =  \frac{dtEnergy(f')}{df'}  
\end{displaymath}
Note that, unlike Lemma~\ref{DynamicOptimizer}, the total energy function $tEnergy(f')$ is not an increasing function of $f'$.  Here we can't say with certainity that $\lambda_1 > 0$.
Applying complementary slackness on the the first constraint, we get $\lambda_1(\frac{\sum_{m=1}^{M}\pi_{m'}w_m}{t_{budget} - \sum_{m=1}^{M}w_mdt_a} -f') = 0$, either $\lambda_1 = 0$ or the first constraint is met with a slack. 

If the unconstrained optimizer of the total energy obtained by $\frac{dtEnergy}{df'} = 0$ meets the deadline constraint with a slack, i.e. the unconstrained optimizer $f'_{unconstrained} > \frac{\sum_{m=1}^{M}\pi_{m'}w_m}{t_{budget} - \sum_{m=1}^{M}w_mdt_a}$, $\lambda_2$ becomes zero. Otherwise, $\lambda_2 =  \frac{dtEnergy}{df'}$ at $f' = \frac{\sum_{m=1}^{M}cw_m}{t_{budget} - \sum_{m=1}^{M}w_mdt_a}$ as in Lemma~\ref{DynamicOptimizer}

Thus the optimal point is:
\begin{displaymath}
\begin{aligned}
&\lambda_1 = max(\frac{ddEnergy(f')}{df'}, 0) \\
&\text{where } f' =  \frac{\sum_{m=1}^{M}\pi_{m'}w_m}{t_{budget} - \sum_{m=1}^{M}w_mdt_a}  \text{ is the minimum allowed reference frequency by the deadline}\\
&\lambda_2 = 0\\
&f' = max(\text{Unconstrained optimizer of }  tEnergy(f'),  \frac{\sum_{m=1}^{M}\pi_{m'}w_m}{t_{budget} - \sum_{m=1}^{M}w_mdt_a})
\end{aligned}
\end{displaymath}
\end{proof}
\end{document}